\documentclass[11pt,letter]{article}
\usepackage{caption}
\usepackage{subcaption}
\usepackage[colorlinks=true,linkcolor=blue,citecolor=blue]{hyperref}
\usepackage{tikz}
\usepackage{tikz-qtree}
\usepackage{amsmath,amssymb}
\usepackage[capitalise]{cleveref}
\usepackage[margin=1in]{geometry}
\usepackage{palatino}
\usepackage{setspace}
\usepackage{amsthm,thmtools, thm-restate}
\usepackage{xcolor} 
\usepackage{comment}
\usepackage{lineno}

\newtheorem{theorem}{Theorem}[section]
\newtheorem{definition}[theorem]{Definition}
\newtheorem{observation}[theorem]{Observation}

\newtheorem{lemma}[theorem]{Lemma}

\newtheorem{corollary}[theorem]{Corollary}

\newcommand{\fand}{\mathsf{AND}}
\newcommand{\for}{\mathsf{OR}}

\newcommand{\cc}{\mathsf{C}}
\newcommand{\sen}{\mathsf{s}}
\newcommand{\bsen}{\mathsf{bs}}
\newcommand{\dqc}{\mathsf{D}}
\newcommand{\dqcs}[1]{\mathsf{D_{#1}}}
\renewcommand{\deg}{\mathsf{deg}}
\newcommand{\TRIBES}{\mathsf{TRIBES}_{n}}
\newcommand{\TRIBESCS}{(\mathsf{TRIBES}_n)_{CS}}

\title{On Condensation of Block Sensitivity, Certificate Complexity and the $\fand$ (and $\for$) Decision Tree Complexity}
\author{
Sai Soumya Nalli\thanks{Microsoft Research Lab, Bangalore, India. \texttt{saisoumya7208@gmail.com}}
\and
Karthikeya Polisetty$^\dagger$
\and
Jayalal Sarma
\thanks{Indian Institute of Technology Madras, Chennai, India. \texttt{cs22b026@smail.iitm.ac.in}, \texttt{jayalal@cse.iitm.ac.in}}
}
\date{}

\begin{document}
\maketitle

\begin{abstract}
Given an $n$-bit Boolean function with a complexity measure (such as block sensitivity, query complexity, etc.) $M(f) = k$, the hardness condensation question asks whether $f$ can be restricted to $O(k)$ variables such that the complexity measure is $\Omega(k)$?

In this work, we study the condensability of block sensitivity, certificate complexity, AND (and OR) query complexity and Fourier sparsity. We show that block sensitivity does not condense under restrictions, unlike sensitivity: there exists a Boolean function $f$ with query complexity $k$ such that any restriction of $f$ to $O(k)$ variables has block sensitivity $O(k^{\frac{2}{3}})$. This answers an open question in \cite{Goos2024} in the negative. The same function yields an analogous incondensable result for certificate complexity. We further show that $\mathsf{AND}$(and $\mathsf{OR}$) decision trees are also incondensable. 
\end{abstract}

\hypersetup{linkcolor=blue}

\tableofcontents

\section{Introduction}

Hardness condensation is a powerful lower-bound technique in Boolean function complexity theory that allows us to reduce the size of a problem while preserving its essential complexity characteristics. In the context of Boolean functions, hardness condensation refers to the process of identifying smaller subfunctions that retain the complexity properties of the original function. Conceptually, it seeks to determine whether the computational hardness is ``concentrated" in specific parts of the input domain.


\vspace{0.3cm}
\noindent

If $f : \{0,1\}^n \to \{0,1\}$ is a Boolean function and $M$ is a complexity measure (like sensitivity, block sensitivity, etc.) with $M(f) = k$, we say that $f$ \emph{condenses under restriction} if there exists a partial assignment $\rho : [n] \rightarrow \{0, 1, *\}$ where $|\rho^{-1}(*)| = k$ variables, such that the resulting $k$-bit function $f' := f|_{\rho} : \{0,1\}^k \to \{0,1\}$ satisfies $M(f') = \Omega(k)$. Moreover, we say that $f$ \emph{condenses losslessly} if $M(f')$ condenses to $\Theta(M(f))$  and \emph{condenses exactly} if it condenses to $M(f)$.

Hardness condensation is used to show lower bounds in circuit complexity \cite{BureshSanthanam06} and proof complexity\cite{ProofComplexityRazborov2016}. Recently, \cite{Goos2024} showed that deterministic query complexity does not admit lossless condensation: they construct a function with query complexity $k$ such that any restriction to $k$ variables has query complexity at most $O(k^{3/4})$


\paragraph{Our Results:}

We give a function that demonstrates the incondensability of block sensitivity, thereby answering the open question posed in \cite{Goos2024}.  

\begin{theorem}
\label{thm:bscccond}
    There exists a function on $k^4$ variables with block sensitivity (resp. certificate complexity) $k^3$ such that any restriction of the function to $O(k^3)$ variables will have block sensitivity (resp. certificate complexity at most $O(k^2)$.    
\end{theorem}

Here the function is a modified form of the Rubinstein's function \cite{Rubinstein1995} that initally proved the quadratic gap between sensitivity and block sensitivity. It constructed as an $\for$ of $k^2$ copies of a base function on $k^2$ variables. We also demonstrate that the extent of incondensability for block sensitivity, as presented above, is optimal for Rubinstein-like functions across all possible parameters. The analysis can be found in the Appendix~\ref{subsec:Optimal_gaps}. 
~\cref{thm:bscccond} has also been independently shown by Kayal {\em et al}~\cite{KMPS05}.

We show incondensability for $\fand$ (and $\for$) decision trees using the function used in \cite{Goos2024} for the incondensability of query complexity. We exploit the connection between these measures and the $0$-decision trees (\cite{and_zero_decision_trees} (see preliminaries for a definition) to prove this. More precisely, we prove the following:
\begin{theorem}
    There exists a boolean function $f$ with $\fand$-query complexity $\dqcs{\wedge} = k$ (resp $\dqcs{\vee}$) such that any restriction of the function to $k$ variables will have $\dqcs{\wedge} \le \tilde{O}(k^{3/4})$.
\end{theorem}

The incondesibility result for $\for$ query complexity follows for essentially  the same proof structure.

Further we also observe that Forurier Sparsity follows a weak form of condensation: for any Boolean function with Fourier sparsity $k$, there exists a restriction to $O(k^2)$ variables such that the Fourier sparsity of the restricted function is $k$. This follows directly from the result on monomial sparsity in \cite{hart2025condensing}.

\section{Preliminaries}
In this section, we will define various complexity measures of Boolean functions that we will be working with. Along with the definitions, we will also discuss some basic properties and known results about these measures that will be useful in the later sections.

For a Boolean function $f : \{0, 1\}^n \mapsto \{0, 1\}$ and $x \in \{0, 1\}^n$:
    The \emph{sensitivity} $\sen(f, x)$ is the number of bits $i$ such that $f(x) \neq f(x \oplus e_i)$.
    The \emph{block sensitivity} $\bsen(f, x)$ is the maximum number of disjoint blocks $B_j \subseteq [n]$ such that $f(x) \neq f(x \oplus e_{B_j})$.
    The \emph{certificate complexity} $\cc(f, x)$ is the size of the smallest $S \subseteq [n]$ such that $f$ is constant on the sub-cube $\{y : y|_{S} = x|_S\}$.

\begin{definition}{\cite{and_zero_decision_trees}}
    Given a deterministic decision-tree $T$ over $\{0, 1\}^p$ , the 0-depth (or 1-depth resp.) of $T$ is the maximum number of queries which are answered 0 (or 1 resp.), in any root-to-leaf path of $T$. The 0-query complexity of a Boolean function $f$, denoted by $\dqcs{0}(f)$, is the smallest 0-depth of $T$, taken over deterministic decision-trees $T$ which compute $f$.
\end{definition}

\begin{definition}[\textbf{AND-Decision Tree}]
    An \textbf{$\fand$-Decision Tree} is a deterministic query model that generalizes the classical decision tree. While a classical tree queries individual input bits $x_i\in \{0, 1\}$, each internal node in an $\mathsf{AND}$-decision tree queries a conjunction of a subset of input bits. Formally, a node $v$ computes a function $f_v(x) = \bigwedge\limits_{i\in S_v} x_i$ for some $S_v\subseteq [n]$, the decision to which child to choose is based on the binary outcome of this query.
\end{definition}

\begin{definition}[\textbf{OR-Decision Tree}]
    An \textbf{$\for$-Decision Tree} is a deterministic query model that generalizes the classical decision tree. While a classical tree queries individual input bits $x_i\in \{0, 1\}$, each internal node in an $\mathsf{OR}$-decision tree queries a disjunction of a subset of input bits. Formally, a node $v$ computes a function $f_v(x) = \bigvee\limits_{i\in S_v} x_i$ for some $S_v\subseteq [n]$, the decision to which child to choose is based on the binary outcome of this query.
\end{definition}

\section{Incondensability of Certificate Complexity \& Block Sensitivity}

Our construction builds on Rubinstein's function~\cite{Rubinstein1995} defined as follows. The function is on $n = k^2$ variables, which are divided into $k$ blocks with $k$ variables each. The value of the function is $1$ if there is at least one block with exactly $k$ consecutive $1$s in it and it is $0$ otherwise.
We define the Modified Rubinstein function $f$ using the function $g$ on $k^2$ variables as a component:

\begin{definition}[\textbf{Modified Rubinstein's Function}]
\label{def:Mod_Rub}
We define the function $f$ using the function $g$ on $k^2$ variables as a component:

\begin{equation}
    g(x) = \begin{cases}
1 & \textrm{ if } \exists j : x_{jk+1}=\cdots=x_{jk+k}=1
\quad \& \quad
x_i=0 \text{ for all } i \notin \{jk+1,\ldots,jk+k\}\\
0 & \textrm{ otherwise }
\end{cases}
\end{equation}

Consider the extension on $k^2$ copies of the following kind: let $X_1, X_2, ..., X_{k^2}$ be disjoint set of variables of size $k^2$ and $x_1,.., x_{k^2}$ are the input blocks on each of this.

We now define $f$ as the OR of $g$ value of all of the blocks.
$$f(x) = \bigvee_{i \in [k^2]}g(x_i)$$

\end{definition}

We analyze the defined modified Rubinstein's function to understand it's certificate complexity. We firstly show that any certificate for the function has to adhere to the following properties.

\begin{observation}
    Consider $x$ such that $g(x) = 0$ and let $B_j = \{jk, \cdots, jk + k -1\}$ be the $j^{th}$ contiguous set of indices, then there exist one of the following:

    \begin{enumerate}
        \item For each $j \in \{0, 1, \cdots k-1\}$, there is an index $i_j \in B_j$ such that $x_{i_j} = 0$ 
        \item There exist 2 $j, j' (j \neq j')$ such there there are indices $i \in B_j$ and $i' \in B_{j'}$ such that $x_i = x_{i'} = 1$
    \end{enumerate}
    \label{obs:0cert_mod_Rub}
\end{observation}

\begin{proof}
    Note that is either one of the above happens we can conclude that $g(x) = 0$.

    Now we need to show that if they both do not happen then $g(x) = 1$. Suppose that that the 2 statements do not hold together, then we know that $\exists j $ such that $\forall i \in B_j, x_{i_j} = 1$. However we also know that there is no other $j'$ such that any index from $B_{j'}$ contains a 1. This implies that $g(x) = 1$
\end{proof}

Using the above observation we now analyze the certificate complexity of $g$:

\begin{description}
    \item[$g(x) = 1$] Note that when this happens all the variables have to revealed to conclude that $g(x)$ is indeed 1. If not any of the variables (say $i$) outside the certificate can contain a different value from $x_i$ and the output is actually a 0. Hence $\cc_1(g) = k^2$

    \item[$g(x) = 0$] The observation \ref{obs:0cert_mod_Rub} allows us to conclude that $\cc(f, x) \le \max(k, 2)$ where the equality happens for the input $\vec{0}$ and hence $\cc_0(g) = \max(k, 2)$
\end{description}

Now we consider the extension on $k^4$ blocks of the following kind: let $X_1, X_2, ..., X_{k^2}$ be contiguous variables of size $k^2$ and $x_1,.., x_{k^2}$ are the input blocks on each of this.

We define $f(x) = \bigvee_{i \in [k^2]}g(x_i)$, in other words the OR of g value of all of the blocks. We examine the certificate complexity

\begin{description}
    \item[$f(x) = 1$] Note that in this case the $x_i$ which gives $g(x_i)$ has to be revealed and we are done. This implies $\cc_1(f) \le k^2$. For any such input $x$, if less than $k^2$ bits are revealed the the output can always be made 0 since no block conclusively has output 1. ($\cc_1(g) = k^2$)

    \item[$f(x) = 0$] In this case we need a $0$ certificate and from the previous observations we know that $\cc_0(f) \le k^2(\max(k, 2))$. Consider $\vec{0}$ if less than $k^3$ values are revealed then there's at least one block of size $k$ in one of the $X_i$ where no values are revealed and can be made to be all 1's to get $\cc(f, \vec{0}) = k^3$
    
\end{description}

\begin{observation}
    Let $\rho$ be a restriction on $g$ ($k^2$ variables) such that $g_{\rho}$ is not a constant function, then the certificate complexity $\cc_0(g_{\rho}) \le \max (2, \frac{r}{k})$ where $r$ is the number of free variables
    \label{obs:cc_partial}
\end{observation}

\begin{proof}
    Since $g_{\rho}$ is not trivially constant we can conclude that $\rho$ is of the following kind:

    \begin{enumerate}
        \item There are no 1's assigned by $\rho$. Note that in this case for any $g_{\rho}(x) = 0$, we have the following cases:
        \begin{itemize}
            \item Suppose some of the inputs bits are set to a 1. In this case we have either (a) One block $X_i$ with both a 0 and a 1 or (b) 2 different blocks with a 1. In this case we know that $\cc(g_{\rho}, x) \le 2$

            \item None of the input bits are set to 0. In this case suppose consider all the $B_i$'s (from continuous blocks of size $k$) that are not assigned any 0 by $\rho$. Suppose one bit from each of them is taken that acts as a 0-certificate and hence $\cc(g_{\rho}, \vec{0}) \le \frac{r}{k}$ since there are at most $\frac{r}{k}$ unassigned complete blocks.
        \end{itemize}
    
        \item Suppose there is one particular block $x_i$ which has some 1's and the remaining bits in it unassigned, then any $x$ such that $g_{\rho}(x) = 0$ has the property that either (a) There is at least one 0 in the block $x_i$ (or) (b) There is at least one 1 in another block. In both cases that bit itself becomes the certificate for us to conclude that $g_{\rho}(x) = 0$. Hence $\cc_0(f_{\rho}, x) = 1$
        
    \end{enumerate}
\end{proof}

\begin{theorem}
    Consider any restriction $\rho$ on $f$ ($f$ has certificate complexity $k^3$) such that there are $O(k^3)$ free variables, then certificate complexity can be at most $O(k^2)$.
    \label{thm:cc_incon}
\end{theorem}

\begin{proof}
    \begin{itemize}
        \item Suppose that $f_{\rho}(x) = 1$, then we know that there exists at least one $g$ such that $g_{\rho}(x) = 1$ and hence $\cc_1(f_{\rho}) \le k^2$

        \item If $f_{\rho}(x) = 0$, then we know that revealing a 0 certificate in each of the $g$'s would suffice. Consider $g_{\rho, i}(x_i)$ we know by \ref{obs:cc_partial} that $\cc_0(g_{\rho, i}) \le \max (2, \frac{r_i}{k}) < \frac{r_i}{k} + 2$ where $r_i$ is the number of free variables in the block $X_i$. 

        \begin{align}
            \cc_0(f_{\rho})  & \le \sum_i \cc_0(g_{\rho, i})\\ & < \sum_i \frac{r_i}{k} + \sum_i 2\\
            &\le \frac{O(k^3)}{k} + 2k^2 = O(k^2)
        \end{align}
        
    \end{itemize}
\end{proof}

\begin{corollary}
    Consider any restriction $\rho$ on $f$ ($f$ has block sensitivity $k^3$) such that there are $O(k^3)$ free variables, then block sensitivity can be at most $O(k^2)$.
    \label{cor:bs_incon}
\end{corollary}

To prove the above, we firstly observe that the $\bsen(f) = \cc(f) = k^3$. Since we know from Theorem \ref{thm:cc_incon} that any restriction on $O(k^3)$ variables has certificate complexity at most $O(k^2)$, the same holds for block sensitivity since we know that for any function $f', \bsen(f') \le \cc(f')$. A direct proof for the incondensibility of block sensitivity along with the analysis of the modified Rubinestein function is given in Appendix \ref{app:bs_incon}

\section{Incondensability of the \texorpdfstring{$\fand$}{fand} (and \texorpdfstring{$\for$}{for}) query complexity}

\cite{and_zero_decision_trees} shows that $\dqcs{\wedge}(f)$ is closely related to another measure 0-decision tree depth $\dqcs0(f)$ which is defined as follows. For a deterministic decision tree $T$ computing the function $f$, the $0$ depth is the maximum number of queries which are answered 0 in any root-to-leaf path in the tree. $\dqcs0(f)$ is the smallest 0 depth over all the decision trees that compute $f$. The following relation was shown:

\[\dqcs0(f) \le \dqcs{\wedge}(f) \le \dqcs0(f) \log(n+1) \]

We use lower bound the $\dqcs{\wedge}$ of $\TRIBES$ to understand the same adversarial argument that gives the decision tree lower bound gives a $n^2 - n + 1$ lower bound.

\[\TRIBES(x) := \bigwedge^n_{i = i}\bigvee^n_{j = 1} x[i, j]\]

\begin{lemma}
    $\dqcs{0}(\TRIBES) = n^2-n+1 $
\end{lemma}

\begin{proof}
Let $T$ be a decision tree that computes $\TRIBES$. We describe an adversarial strategy, i.e., the values returned by the responder so that the there is at least one path from the root to the leaf such that there are $n^2-n + 1$ 0's queried on the path. The adversarial strategy is as follows: Initialize the empty sets: \(X_1, X_2, ...X_n\) which correspond to queried variables from the sets that the $\for$ is performed over i.e \(X_{i} = \{x[i, j] \mid x[i,j] \quad \text{is queried}, \quad j \in [n] \}\). If any $x[i, j]$ is queried by the $T$, add it to $X_i$ and if \(|X_i| = n \text{ and }\sum_i|X_i| < n^2 \) then return 1 else return 0.

We claim that with the above strategy forces the $T$ to query for all the variables before outputting a value. This would directly imply $n^2- n + 1$ 0's queried since that many queries (all but the last queried variables in each block + the last queried variable). Suppose for the sake of contradiction, some variables have not been queried. We can observe that the strategy has assigned variables such that if all remaining variables are assigned (i)1 then that assignment would lead to $\TRIBES(x) = 1$ and (ii)0 would lead to $\TRIBES(x) = 0$. Thus this leads to a contradiction.
\end{proof}

A matching upper bound through the following strategy leads to $\dqcs{\wedge}(\TRIBES) = n^2- n+1$: 
\begin{enumerate}
    \item Query $n-1$ variables in each block $X_i$.
    \item Consider all the blocks $X_{i_1}, ..., X_{i_t}$ where all the variables have been set to 0. Now query the remaining variables in these as an $\bigwedge^t_{s = 1}x[i_s, r_{i_s}]$ where $r_i$ was the single unqueried variable in $X_i$. Return the output.
\end{enumerate} 

\paragraph{$\fand$ query complexity of the cheat sheet}

Consider the cheat sheet function \(f := \TRIBESCS\) mapping \( \left( \{0,1\}^{n^2}\right)^c \times \left(\{0,1\}^{c\cdot m}\right)^{2^c} \rightarrow \{0,1\}\) where
\(c := 10\log(n^2), m := n\log(n^2)\). Note that $c$ is set to \(10 \log \dqc(\TRIBES)\) while \(m := C(\TRIBES) \log \dqc(\TRIBES)\). Following the proof\cite{Scott2015} that shows that $g_{CS}$ has deterministic query complexity $\Omega(\dqc(f))$ where $f$ was the original function we show the below:

\begin{theorem}
    \(\dqcs{0}(\TRIBESCS) = \Omega(n^2)\)
\end{theorem}

\begin{proof}
    Consider the following strategy: As long as the bit is outside the input copies \(x_1, x_2, \cdots x_c\) return 0. For bits in an $x_i$ consider the adversarial strategy we chose for the individual \(\TRIBES\) function such that only computing all bits can lead to $\TRIBES(x_i)$. Suppose for the sake a contradiction there is a deterministic decision tree $T$ that outputs the answer without querying $n^2$ bits on the above strategy. Note that this implies (i) There is a cheatsheet cell \(\in \{0,1\}^{c\cdot m}\) that has not be queried at all (even for a single bit). Since even querying 1 bit of every cheatsheet cell implies $2^c \ge n^2$ queries, (ii)There is no $x_i$ such that it is fully queried.

    Note that the above means that the remaining bits of each $x_i$ can be set that $f(x_1)f(x_2) \cdots f(x_c)$ points towards the completely unqueried cheatsheet cell. Hence the output can be set to either 0 or 1 by changing this cheatsheet cell.

    We proved that $T$ queries at least $n^2$ queries when it comes across an input designed by the strategy and note that there can only be $n$ 1's that would be outputted (when $n$ full blocks have been queried across different copies). Hence we can conclude that $\dqcs{0}(\TRIBESCS) = \Omega(n^2 -n)$ and therefore $\dqcs{\wedge}(\TRIBESCS) = \Omega(n^2)$
    
\end{proof}

We also know that $\dqcs{\wedge}(\TRIBESCS) \le \dqc(\TRIBESCS) = \tilde{O}(n^2)$. \cite{Goos2024} show that restriction on $\tilde{O}(n^2)$ always leads to $\dqc(f \mid_{\rho}) \le \tilde{O}(n^{3/2})$. This directly implies that on restricting on (possibly) lesser variables would lead to $\dqc(f \mid_{\rho}) \le \tilde{O}(n^{3/2})$ and hence $\dqcs{\wedge}(f \mid_{\rho}) \le \tilde{O}(n^{3/2})$ resulting in incondensability of $\dqcs{\wedge}$

\begin{corollary}
    For the function $f': = \bigvee^n_{i = 1}\bigwedge^n_{j = i} x[i, j]$ the $\for$ query complexity will be incondensable
\end{corollary}

The above follows through a similar symmetric argument on the above function. For this we use the following $\mathsf{OR}$ counterpart of the result from \cite{and_zero_decision_trees} : 
$\dqcs1(f) \le \dqcs{\vee}(f)$

\section{Discussion and Conclusion}

In this work, we have explored the condensability and incondensability properties of several Boolean function complexity measures. We establish the incondensability of the block sensitivity (and certificate complexity), thus resolving an open question posed in \cite{Goos2024}. We further demonstrated that this incondensability bound is tight for Rubinstein-like functions. Additionally, we have shown that $\mathsf{AND}$ and $\mathsf{OR}$ decision tree complexities are incondensable. We also presented a weak condensation result for Fourier sparsity.

\paragraph{Connection between Condensability and Separations:}
The connections between condensability and the separations between various Boolean function complexity measures is worth noting. 
In particular, 
    for any two Boolean function complexity measures $M_1$ and $M_2$, if $M_1$ can be condensed losslessly and if $M_2 \in \Theta(M_1)$ then $M_2$ can also be condensed losslessly.

The above observation gives us the necessity of separations with condensable measures for incondensability but not sufficiency. This can be observed in the incondensability shown in the previous sections. The \textsc{Tribes}$_n$ function exhibits a separation between query complexity $\dqc(f)$ with degree $\deg$ and sensitivity $\sen$. Similarly, the Rubinstein-like function shows a separation between block sensitivity $\bsen$ and sensitivity $\sen$. However, both degree and sensitivity can be condensed losslessly. We can also see that $\bsen, \cc$ condense exactly for Monotone functions since $\sen(f) = \bsen(f) = \cc(f)$\cite{Nisan1989} for a monotone function $f$. 
In particular, this gives, 
    If there exists a Boolean function $f$ such that there is no restriction $f_{\rho}$ on $O(\sen(f))$ where $\bsen(f_{\rho}) = \Omega(\sqrt{\bsen(f)})$, then $f$ has a super quadratic gap between $\bsen(f), \sen(f)$. 
This follows by contradiction, suppose for every function $f$, $\bsen(f)  = O(\sen(f)^2)$ then we have a restriction on the number of variables $\sen(f)$ such that $\sen(f_{\rho}) = \sen(f) \ge c \sqrt{\bsen(f)}$ and hence $\bsen(f_{\rho}) = \Omega(\sqrt{\bsen(f)})$.

\paragraph{Lower Bounds for Condensation}
The Sensitivity Theorem \cite{huang2019inducedsubgraphshypercubesproof} implies that for any Boolean function, there exists a restriction to $\sen(f)$ variables (which is $O(\bsen(f))$) such that the block sensitivity of the restricted function is at least $\Omega(\bsen(f)^{1/4})$. A stronger bound, due to Kayal {\em et al.} \cite{KMPS05}, shows that there exists a restriction to $\bsen(f)$ variables for which the block sensitivity is at least $\Omega(\bsen(f)^{1/2})$. An interesting question is whether such a fraction can be achieved by some Boolean function, or whether the fraction can be improved.

\bibliographystyle{alpha}
\bibliography{references}

@misc{KMPS05,
author = {Chandrima Kayal and Rajat Mittal and Manaswi Paraashar and Nitin Saurabh},
note = "private communication",
year = 2025
}

@article{hart2025condensing,
  title={Condensing Hardness in Boolean Functions},
  author={Hart, Gabriel},
  year={2025},
}

@InProceedings{and_zero_decision_trees,
  author =	{Loff, Bruno and Mukhopadhyay, Sagnik},
  title =	{{Lifting Theorems for Equality}},
  booktitle =	{36th International Symposium on Theoretical Aspects of Computer Science (STACS 2019)},
  pages =	{50:1--50:19},
  series =	{Leibniz International Proceedings in Informatics (LIPIcs)},
  ISBN =	{978-3-95977-100-9},
  ISSN =	{1868-8969},
  year =	{2019},
  volume =	{126},
  editor =	{Niedermeier, Rolf and Paul, Christophe},
  publisher =	{Schloss Dagstuhl -- Leibniz-Zentrum f{\"u}r Informatik},
  address =	{Dagstuhl, Germany},
  URL =		{https://drops.dagstuhl.de/entities/document/10.4230/LIPIcs.STACS.2019.50},
  URN =		{urn:nbn:de:0030-drops-102892},
  doi =		{10.4230/LIPIcs.STACS.2019.50}
}

@misc{huang2019inducedsubgraphshypercubesproof,
      title={Induced subgraphs of hypercubes and a proof of the Sensitivity Conjecture}, 
      author={Hao Huang},
      year={2019},
      eprint={1907.00847},
      archivePrefix={arXiv},
      primaryClass={math.CO},
      url={https://arxiv.org/abs/1907.00847}, 
}

@article{Goos2024,
  title={Hardness Condensation by Restriction},
  author={Mika G{\"o}{\"o}s and Ilan Newman and Artur Riazanov and Dmitry Sokolov},
  journal={Proceedings of the 56th Annual ACM Symposium on Theory of Computing},
  year={2024},
  url={https://api.semanticscholar.org/CorpusID:266903666}
}

@article{Scott2015,
  title={Separations in query complexity using cheat sheets},
  author={Scott Aaronson and Shalev Ben-David and Robin Kothari},
  journal={Proceedings of the forty-eighth annual ACM symposium on Theory of Computing},
  year={2015},
  url={https://api.semanticscholar.org/CorpusID:4865642}
}

@article{Rubinstein1995,
  title = {Sensitivity vs. block sensitivity of boolean functions},
  author = {Rubinstein, D.},
  journal = {Combinatorica},
  volume = {15},
  number = {2},
  year={1995},
}

@inproceedings{Nisan1989,
author = {Nisan, N.},
title = {CREW PRAMS and decision trees},
year = {1989},
isbn = {0897913078},
publisher = {Association for Computing Machinery},
address = {New York, NY, USA},
url = {https://doi.org/10.1145/73007.73038},
doi = {10.1145/73007.73038},
booktitle = {Proceedings of the Twenty-First Annual ACM Symposium on Theory of Computing},
pages = {327–335},
numpages = {9},
location = {Seattle, Washington, USA},
series = {STOC '89}
}

@INPROCEEDINGS{BureshSanthanam06,
  author={Buresh-Oppenheim, J. and Santhanam, R.},
  booktitle={21st Annual IEEE Conference on Computational Complexity (CCC'06)}, 
  title={Making hard problems harder}, 
  year={2006},
  volume={},
  number={},
  pages={15 pp.-87},
  doi={10.1109/CCC.2006.26}}

@article{ProofComplexityRazborov2016,
author = {Razborov, Alexander},
title = {A New Kind of Tradeoffs in Propositional Proof Complexity},
year = {2016},
issue_date = {May 2016},
publisher = {Association for Computing Machinery},
address = {New York, NY, USA},
volume = {63},
number = {2},
issn = {0004-5411},
url = {https://doi.org/10.1145/2858790},
doi = {10.1145/2858790},
journal = {J. ACM},
month = apr,
articleno = {16},
numpages = {14}
}

\appendix

\section{Direct Proof of Incondensibility of Block Sensitivity}\label{app:bs_incon}

We analyze the sensitivity and block sensitivity of $f$ and $g$ defined in Def \ref{def:Mod_Rub}:

\begin{description}
    \item[$g(x) = 1$] Note that for such an input all bits are sensitive and hence we have $\sen(g, x) = \bsen(g, x) = k^2$
    \item[$g(x) = 0$] We observe that any such input can have at most one neighbor with $g$ value 1. Hence $\sen(g, x) \le 1$. However, at most $k$ inputs have $g$ value 1 and they can all be reached by flipping $k$ sized blocks in $\Vec{0}$. Hence, $\bsen(g, x) \le k $ and for the all 0 inputs equality happens.
\end{description}

\begin{description}
    \item[$f(x) = 1$] 
    \begin{itemize}
        \item Note that an input such that more than 1 sub function has $g(x_i) = 1$ then $\sen(f, x) = 0$. If there is exactly one $i$ such that $g(x_i) = 1$ then we know that only flipping bits in this block changes the value to give $\sen(f, x) = k^2$
        \item We know that if a subfunction on $x_i$ gives $g(x_i) = 1$ then any sensitive block $b_i$ should have at least one bit in this block, i.e $X_i \cap B_j \neq \phi$ and hence there can be at most $k^2$ disjoint blocks i.e $\bsen(f, x) \le k^2$ The equality happens for the same input described in the previous input.
    \end{itemize}
    \item[$f(x) = 0$]
    \begin{itemize}
        \item Each subfunction has at most one neighbor in the dimensions of $x_i$ such that $g(x_i) = 1$. Hence $\sen(f,x) \le k^2$ 
        \item Suppose that there is a sensitive block system $\mathcal{B} = \bigcup_i B_i$ for $x$. Let us assign an $X_{i'}$ for every block based on the subfunction that it flips to a 1. Consider the disjoint blocks $B_{i1}, .., B_{in_i}$ that were assigned $X_i$ - since we know that these when restricted to $X_i$ form sensitive blocks for $g(x_i)$ and $\bsen_0(g) \le k$ which implies $n_i \le k$.

        By summing this up over all the blocks assigned over all the $X_i$'s we have

        $$\bsen(f, x) \le \sum_{i \in [k^2]} \bsen_0(g_i) = k^3$$
        
        We note that this happens for the all 0 input and hence $\bsen_0(f) = k^3$
    \end{itemize}
\end{description}

\begin{lemma}
    Let $\rho$ be a restriction on $g$ ($k^2$ variables) such that $g_{\rho}$ is not a constant function, then the block sensitivity $\bsen_0(g_{\rho}) \le \max (1, \frac{r}{k})$ where $r$ is the number of free variables
    \label{obs:bs_partial}
\end{lemma}

\begin{proof}
    Since $g_{\rho}$ is not trivially constant we can conclude that $\rho$ is of the following kind:

    \begin{enumerate}
        \item There are no 1's assigned by $\rho$. Note that in this case for any $g_{\rho}(x) = 0$, the only possible inputs that give $g_{\rho}(x) = 1$ is for the $b \le \frac{r}{k}$ number of inputs which have contiguous k sized blocks. Note that this is achieved for $\Vec{0}$ with $\bsen(f_{\rho}, x) \le \frac{r}{k}$
        \item Suppose there is one particular block $x_i$ which has some 1's and the remaining bits in it unassigned, note that the output can only be changed to a 1 when all remaining the bits in this $x_i$ are made to 1. Hence there can be only 1 sensitive block since any other disjoint block can not be sensitive. $\bsen(f_{\rho}, x) \le 1$
    \end{enumerate}
\end{proof}

\begin{theorem}
    Consider any restriction $\rho$ on $f$ ($f$ has block sensitivity $k^3$) such that there are $O(k^3)$ free variables, then block sensitivity can be at most $O(k^2)$.
    \label{thm:bs_incon}
\end{theorem}

\begin{proof}
    For any $f_{\rho}(x) = 1$ as we have examined before if a subfunction on $x_i$ gives $g_{\rho}(x_i) = 1$ then any sensitive block $b_i$ should have at least one bit in this block, i.e $X_i \cap B_j \neq \phi$ and hence there can be at most $k^2$ disjoint blocks i.e $\bsen(f, x) \le k^2$

    Now let us suppose $f_{\rho}(x) = 0$. If the restriction directly gives a constant function then $\bsen(f_{\rho}) = 0$ so we need to deal with the situation where we have a non constant restriction. Suppose that $r_i$ are the number of free variables in $X_i$, note that a sensitive block has to change at least one of the subfunctions to a 1 and this implies a disjoint sensitive block system $\mathcal{B} = \bigcup_i B_i$. Let us assign an $X_{i'}$ for every block based on the subfunction that it flips to a 1. Consider the disjoint blocks $B_{i1}, .., B_{in_i}$ that were assigned $X_i$ - since we know that these when restricted to $X_i$ for sensitive blocks from obs \ref{obs:bs_partial} we can conclude that $n_i \le \frac{r_i}{k} + 1$.

    By summing this up over all the blocks assigned over all the $X_i$'s we have


    \begin{align}
        \bsen(f_{\rho})  & \le \sum_i \bsen(g_{\rho, i})\\ & < \sum_i \frac{r_i}{k} + \sum_i 1\\
        &\le \frac{O(k^3)}{k} + k^2 = O(k^2)
    \end{align}
\end{proof}

\section{Optimality of Incondensability Results (for Rubinstein like functions)}
\label{subsec:Optimal_gaps}
In this section we show that the incondensability result shown in \Cref{thm:bs_incon} is optimal for functions similar to Rubinstein's function. We observe that Rubinstein's function have 3 features which can be parameterized as follows:
\begin{itemize}
    \item Block Size of $g$ : $b$
    \item Number of blocks of $g$ : $n$
    \item Number of repetitions of $g$ in $f$ : $r$
    \item We define optimality to be the separation between $\bsen(f)$ and the maximum possible $\bsen(f_{\rho})$ for any restriction $\rho$ with $O(\bsen(f))$ free variables. More precisely the optimality is quantified as follows: $\log_{\bsen(f_{\rho})}(\bsen(f))$
\end{itemize}

For the function $f$ defined in \Cref{thm:bs_incon} we have $b = k, n = k$ and $r = k^2$. We now show that for any function with these parameters the incondensability result is optimal i.e., $\log_{\bsen(f_{\rho})}(\bsen(f)) = \log_{k^2}(k^3) = \frac{3}{2}$.

\vspace{0.3cm}
\noindent
We define the function $g$ on $b n$ variables as follows:
\begin{equation}
    g(x) = \begin{cases}
1 & \textrm{ if } \exists j : x_{jb + 1} = ... =x_{jb + b} = 1 \textrm{and } x_{i} = 0 \quad \text{for every other index}\\
0 & \textrm{ otherwise }
\end{cases}
\end{equation}

We analyze the sensitivity and block sensitivity as follows:

\begin{description}
    \item[$g(x) = 1$] Note that for such an input all bits are sensitive and hence we have $\sen(f, x) = \bsen(g, x) = b n$
    \item[$g(x) = 0$] We observe that any such input can have at most one neighbor with $g$ value 1. Hence $\sen(g, x) \le 1$. However at most $k$ inputs have $g$ value 1 and they can all be reached by flipping $b$ sized blocks in $\Vec{0}$. Hence $\bsen(g, x) \le n $ and for $\Vec{0}$ this is equal.
\end{description}

Consider the extension on $r b  n$ blocks of the following kind: let $X_1, X_2, ..., X_{r}$ be contiguous variables of size $b  n$ and $x_1,.., x_{r}$ are the input blocks on each of this.

We define $f(x) = \bigvee_{i \in [r]}g(x_i)$, in other words the OR of g value of all of the blocks. We examine the sensitivity and block sensitivity:

\begin{description}
    \item[$f(x) = 1$] 
    \begin{itemize}
        \item Note that an input such that more than 1 sub function has $g(x_i) = 1$ then $\sen(f, x) = 0$. If there is exactly one $i$ such that $g(x_i) = 1$ then we know that only flipping bits in this block changes the value to give $\sen(f, x) = b  n$
        \item We know that if a subfunction on $x_i$ gives $g(x_i) = 1$ then any sensitive block $b_i$ should have at least one bit in this block, i.e $X_i \cap B_j \neq \phi$ and hence there can be at most $b  n$ disjoint blocks i.e $\bsen(f, x) \le b  n$ The equality happens for the same input described in the previous analysis.
    \end{itemize}
    \item[$f(x) = 0$]
    \begin{itemize}
        \item Each subfunction has at most one neighbor in the dimensions of $x_i$ such that $g(x_i) = 1$. Hence $\sen(f,x) \le r$ 
        \item Suppose that there is a sensitive block system $\mathcal{B} = \bigcup_i B_i$ for $x$. Let us assign an $X_{i'}$ for every block based on the subfunction that it flips to a 1. Consider the disjoint blocks $B_{i1}, .., B_{i_{n_i}}$ that were assigned $X_i$ - since we know that these when restricted to $X_i$ form sensitive blocks for $g(x_i)$ and $\bsen_0(g) \le n$ which implies $n_i \le n$.

        By summing this up over all the blocks assigned over all the $X_i$'s we have

        $$\bsen(f, x) \le \sum_{i \in [r]} \bsen_0(f) = r  n$$

        We note that this happens for $(\vec{0})$ and hence $\bsen_0(f) = r  n$
    \end{itemize}
\end{description}

\begin{observation}
    Let $\rho$ be a restriction on $g$ ($b  n$ variables) such that $g_{\rho}$ is not a constant function the block sensitivity $\bsen_0(g_{\rho}) \le \max (1, \frac{m}{b})$ where $m$ is the number of free variables
    \label{obs:bs_gen_partial}
\end{observation}

\begin{proof}
    Since $g_{\rho}$ is not trivially constant we can conclude that $\rho$ is of the following kind:

    \begin{enumerate}
        \item There are no 1's assigned by $\rho$. Note that in this case for any $g_{\rho}(x) = 0$, the only possible inputs that give $g_{\rho}(x) = 1$ is for the $k \le \frac{m}{b}$ number of inputs which have contiguous $b$ sized blocks. Note that this is achieved for $\Vec{0}$ with $\bsen(g_{\rho}, x) \le \frac{m}{b}$
        \item Suppose there is one particular block $x_i$ which has some 1's and the remaining bits in it unassigned, note that the output can only be changed to a 1 when all remaining the bits in this $x_i$ are made to 1. Hence there can be only 1 sensitive block since any other disjoint block cannot be sensitive. $\bsen(g_{\rho}, x) \le 1$
    \end{enumerate}
\end{proof}

\begin{theorem}
    Consider any restriction $\rho$ on $f$ ($f$ has block sensitivity $m = \max{(b n, r n)}$) such that there are $O(m)$ free variables, then block sensitivity can be at most $O(\max(r + \frac{m}{b}, b n))$.
    \label{thm:bs_gen_incon}
\end{theorem}

\begin{proof}
    For any $f_{\rho}(x) = 1$ as we have examined before if a subfunction on $x_i$ gives $g_{\rho}(x_i) = 1$ then any sensitive block $b_i$ should have at least one bit in this block, i.e $X_i \cap B_j \neq \phi$ and hence there can be at most $b  n$ disjoint blocks i.e $\bsen(f, x) \le b  n$

    Now let us suppose $f_{\rho}(x) = 0$. If the restriction directly gives a constant function then $\bsen(f_{\rho}) = 0$ so we need to deal with the situation where we have a non constant restriction. Suppose that $r_i$ are the number of free variables in $X_i$, note that a sensitive block has to change at least one of the subfunctions to a 1 and this implies a disjoint sensitive block system $\mathcal{B} = \bigcup_i B_i$. Let us assign an $X_{i'}$ for every block based on the subfunction that it flips to a 1. Consider the disjoint blocks $B_{i1}, .., B_{i_{n_i}}$ that were assigned $X_i$ - since we know that these when restricted to $X_i$ for sensitive blocks from obs \ref{obs:bs_gen_partial} we can conclude that $n_i \le \frac{r_i}{b} + 1$.

    By summing this up over all the blocks assigned over all the $X_i$'s we have

    $$\bsen(f, x) \le \sum \frac{r_i}{b} + r = \frac{\sum r_i}{b} + r = \frac{O(m)}{b} + r = O(\frac{m}{b} + r)$$

    Therefore, combining both the cases we have $\bsen(f_{\rho}) \le O(\max(r + \frac{m}{b}, b n))$
\end{proof}

For our analysis let $b = n^{\alpha}$ and $r = n^{\beta}$ where $\alpha, \beta \ge 0$. Therefore, our quantification of optimality becomes: $\log_{\bsen(f_{\rho})}(\bsen(f)) = \frac{\max{(\alpha+1, \beta+1)}}{\max{(1, \beta, \alpha+1, \beta + 1 - \alpha)}}$

With some case analysis we can see that the maximum value of this expression is $\frac{3}{2}$ which happens when $\alpha = 1$ and $\beta = 2$ which is exactly the case in \Cref{thm:bs_incon}.

\end{document}